\documentclass[11pt]{amsart}   
\usepackage{amsmath,amssymb} 
\usepackage{graphicx}
\usepackage[font=small,labelfont=bf]{caption}

\usepackage{graphicx}  
\usepackage{dcolumn}   
\usepackage{bm}        
\usepackage{bbm}  
\usepackage{amssymb}   
\usepackage{amsmath}
\usepackage{tikz}
\usepackage{mathrsfs}

\usetikzlibrary{matrix}

\newtheorem{thm}{Theorem}[section]
\newtheorem{lem}[thm]{Lemma}
\newtheorem{prop}[thm]{Proposition}

\newtheorem{defn}[thm]{Definition}
\newtheorem{remark}{Remark}

\numberwithin{equation}{section}

\newcommand{\Z}{{\mathbb Z}} 

\newcommand{\R}{{\mathbb R}}
\newcommand{\C}{{\mathbb C}}

\newcommand{\T}{{\mathbb T}}

\newcommand{\fixmehidden}[1]{}

\def\cB{{\mathcal B}}

\def\cL{{\mathcal L}}
\def\cM{{\mathcal M}}
\def\cN{{\mathcal N}}
\def\cR{{\mathcal R}}

\newcommand{\Vol}{\operatorname{vol}}

\newcommand{\ave}[1]{\left\langle#1\right\rangle} 

\newcommand{\beq}{\begin{equation}}
\newcommand{\eeq}{\end{equation}}

\newcommand{\re}{\operatorname{Re}}

\begin{document} 

\begin{abstract}  
Whereas much work in the mathematical literature on quantum chaos has focused on phenomena such as quantum ergodicity and scarring, relatively little is known at the rigorous level about the existence of eigenfunctions whose morphology is more complex. 

Quantum systems whose dynamics is intermediate between certain regimes -- for example, at the transition between Anderson localized and delocalized eigenfunctions, or in systems whose classical dynamics  is intermediate between integrability and chaos -- have been conjectured in the physics literature to have eigenfunctions exhibiting multifractal, self-similar structure. To-date, no rigorous mathematical results have been obtained about systems of this kind in the context of quantum chaos.   

We give here the first rigorous proof of the existence of multifractal eigenfunctions for a widely studied class of intermediate quantum systems. Specifically, we derive an analytical formula for the Renyi entropy associated with the eigenfunctions of arithmetic \u{S}eba billiards, in the semiclassical limit, as the associated eigenvalues tend to infinity. 

We also prove multifractality of the ground state for more general, non-arithmetic billiards and show that the fractal exponent in this regime satisfies a symmetry relation, similar to the one predicted in the physics literature, by establishing a connection with the functional equation for Epstein's zeta function.
\end{abstract}

\title[Multifractal eigenfunctions]{Multifractal eigenfunctions \\ for a singular quantum billiard}
\author{Jonathan P. Keating}
\address{Mathematical Institute,
University of Oxford,
Andrew Wiles Building,
Radcliffe Observatory Quarter,
Woodstock Road,
Oxford
OX2 6GG, UK.}
\email{jon.keating@maths.ox.ac.uk}
\author{Henrik Uebersch\"ar}
\address{Sorbonne Universit\'e and Universit\'e de Paris, CNRS, IMJ-PRG, F-75006 Paris.}
\email{henrik.ueberschar@cnrs.fr}

\thanks{}
\date{\today} 
\maketitle 

\section{Introduction}

Following the proof of the quantum ergodicity theorem by Snirelman, Zelditch and Colin de Verdi\`ere \cite{Sn74,Z87,CdV85}, one of the central questions in the quantum chaos literature over the past fourty years has been the classification of semiclassical measures which may arise in relation to the eigenfunctions of quantized chaotic systems. Rudnick and Sarnak proposed their Quantum Unique Ergodicity (QUE) conjecture in 1994 \cite{RS94}, asserting that the unique semiclassical measure which should arise along the sequence of Laplace eigenfunctions on manifolds with negative curvature is the Liouville measure. This conjecture was proved by Lindenstrauss and Soundararajan \cite{L06,So10} in the case of arithmetic hyperbolic surfaces. Holowinsky and Soundararajan proved a holomorphic analogue of QUE for Hecke eigenforms \cite{Ho10a,Ho10b}. Anantharaman \cite{A08} generalized this work by ruling out the existence of semiclassical measures localized on a union of points and geodesic segments for Anosov manifolds. On the other hand, Hassell showed that the stadium billiard, although ergodic, does not satisfy QUE \cite{Ha10}. The existence of scars was proved by Faure, Nonnenmacher and de Bi\`evre for hyperbolic toral automorphisms of minimal periods \cite{FNdB03}.

Another class of systems which has attracted considerable attention in the quantum chaos literature are those termed pseudointegrable \cite{RB81}. Such systems are classically close to integrable, yet their quantum dynamics displays features similar to those found in chaotic systems. In particular, their dynamics in phase space is not constrained to tori, but more generally to handled spheres.

Examples of pseudointegrable billiards are rational polygonal billiards. More generally, one sometimes includes systems which can be understood as toy models of pseudointegrable systems, such as quantum star graphs \cite{KMW03}, parabolic toral automorphisms \cite{MR00} and \u{S}eba billiards \cite{Se90}. The latter is a rectangular billiard with a Dirac mass placed in its interior. This system can be rigorously studied by means of von Neumann's theory of self-adjoint extensions. Over the past twenty years much work has been devoted to the study of the spectrum and eigenfunctions of such billiards \cite{BoLeSch00,BeBoK01,BK99,BKW03,BKW04,KMW03,RU12,KU14,KU17,KRo17,KLeRo20}. In particular, arithmetic \u{S}eba billiards have been shown to be quantum ergodic \cite{RU12,KU14}, whereas diophantine \u{S}eba billiards possess scarred semiclassical measures \cite{KU17,KRo17}.

While phenomena such as quantum ergodicity and scarring have attracted a good deal of attention in the mathematical literature on quantum chaos, it is important to point out that many quantum systems possess eigenfunctions whose morphology is far more complex than is simply characterised by their being purely localized or ergodic. For example,  in a wide range of quantum systems in which the spectral statistics sit at the transition between Poisson and random-matrix, it has been argued that the eigenfunctions should exhibit {\em multifractal} self-similar structure. Pseudointegrable systems are one such class of examples, because their dynamics is intermediate between integrability and chaos.

Multifractality is a phenomenon which is observed in many scientific fields ranging from financial mathematics and biology to quantum physics. It describes systems whose self-similarity in a certain scaling regime is so complex that it cannot be described by a single fractal exponent but only by a spectrum of exponents. The multifractality of eigenfunctions is one of the most important phenomena in the theoretical and experimental study of intermediate quantum systems. Physicists have conjectured its existence in a large of class of systems such as the Anderson model \cite{Al86,SG91,Mi00}, pseudointegrable systems \cite{Gi20}, quantum spin chains \cite{AtBo12}, quantum maps \cite{Gi10} and in certain ensembles of random matrices \cite{KM97,BoGi11}, based on heuristic arguments and extensive numerical investigations. Moreover, they have conjectured that the fractal exponents $D_q$ should satisfy a symmetry relation \cite{Gi10} with respect to the critical parameter value $q=\frac{1}{4}$. For an introduction to the extensive physics literature on this subject, see \cite{BHK19}.
 
Although multifractality is an extremely active field in physics, and theoretical evidence for the multifractality of quantum eigenfunctions has been provided by many works in the physics literature as well as physical and numerical experiments, no rigorous mathematical proof is known to date for any class of intermediate systems.

In this paper we prove the multifractality of the eigenfunctions for \u{S}eba's billiard in the strong coupling regime. We derive an analytical formula for the fractal exponents $D_q$. We also prove multifractality of the ground state and show that the symmetry relation of the fractal exponents around the critical exponent $q=\frac{1}{4}$ is a consequence of the functional equation of the Epstein zeta function.

{\bf Acknowledgements:} This work was accomplished during several visits to the Mathematical Institute at the University of Oxford. H.U. would like to thank the institute and particularly Jon Keating's group for their kind hospitality during his stay. The authors would like to thank Nalini Anantharaman, Jens Marklof and Zeev Rudnick for many helpful comments which led to the improvement of this paper.

H.U. gratefully acknowledges the support of the CNRS by means of a ``d\'el\'egation CNRS''. In addition, H.U. was supported by the grant ANR-17-CE40-0011-01 of the French National
Research Agency ANR (project SpInQS).  JPK is pleased to acknowledge support from ERC Advanced Grant 740900 (LogCorRM).  
 


\section{Background and statement of the results}
\subsection{Mathematical formulation of the problem} 
Denote by $(X,\mu)$ a measure space. In classical multifractal analysis one usually evaluates $\mu$ on a box partition $\cB_r$ of $X$, where $r$ denotes the side length of the boxes. To characterise the self-similarity of the measure $\mu$ one studies the behaviour of the moment sum
$$M_q(r)=\sum_{B\in \cB_r} \mu(B)^q$$ in the limit as $r\to 0$. One expects a scaling law of the form $M_q(r)\sim r^{D_q}$, where $D_q$ denotes the fractal exponent. A key goal of multifractal analysis is to determine $D_q$ as a function of $q$.

Let us consider the Laplace-Beltrami operator $-\Delta$ on a two-dimensional Euclidean compact manifold $\cM$ with or without boundary. An eigenfunction $\psi_\lambda$ of $-\Delta$, satisfying $(\Delta+\lambda)\psi_\lambda=0$, may be represented as a superposition of plane waves
\begin{equation}\label{Fourier}
\psi_\lambda(x)=\sum_{\xi\in\cL}\hat{\psi}_\lambda(\xi)e^{i\xi\cdot x}.
\end{equation} 
This may be achieved by embedding the Euclidean billiard in a rectangular enclosure $\cR$ and representing $\psi_\lambda$ with respect to the canonical orthonormal basis of complex exponentials $e^{i \xi\cdot x}$, where $\xi\in\cL$ and $\cL$ is a rectangular lattice. We note that the choice of rectangle, and therefore of the lattice $\cL$, ought to translate into a scaling factor which will appear in lower order terms, because, in order to calculate the fractal exponent $D_q$, one considers the asymptotic behaviour of $\log M_q(r)$. A rigorous formulation of the problem of multifractal eigenfunctions on Euclidean billiards is beyond the scope of this introduction and will be addressed in a separate article. We stress that for the case of \u{S}eba billiards no such ambiguity with respect to the choice of lattice exists.

Let us introduce the probability measure $$\mu_\lambda(\xi)=|\hat{\psi}_\lambda(\xi)|^2$$ where we assume that $\psi$ is $L^2$-normalized.

We consider the moment sum 
\begin{equation}\label{moment_sum}
M_q(\lambda)=\sum_{\xi\in\cL}\mu_\lambda(\xi)^{q}.
\end{equation}

To understand how to determine the scaling parameter which plays the role of $r$, one must distinguish different regimes. 
We will focus here first on the semiclassical regime ($\lambda\to+\infty$).  However, one may also study multifractality of the ground state ($\lambda\to 0$), which we address in section \ref{ground}.
 
\subsection{Multifractality in the semiclassical regime}
To motivate the definition of the fractal exponent associated with the moment sum \eqref{moment_sum}, let us first of all consider the case of a finite-dimensional Hilbert space. Each eigenstate may be represented by a vector $v\in\C^d$, where $d$ denotes the dimension of the space. Let us normalize such that the largest coordinate equals $1$ ($\ell^\infty$ normalization).

The fractal exponent $D_q$ describes how the moment sum associated with $v$ scales with respect to the dimension $d$:
$$\sum_{k=1}^d v_k^{2q}\sim d^{D_q}.$$
If $v$ is flat, then all entries equal $1$ so that $D_q=1$. On the other hand, if $v$ is maximally localized, meaning that one entry equals $1$ and all others vanish, then $D_q=0$. 

Often one normalizes with respect to the $\ell^2$-norm. In this case one observes a decay of the type $\sim d^{(1-q)D_q}$. The fractal exponent $D_q$ may be defined in terms of the Renyi entropy of the probability measure $p_i=v_i^2/\|v\|^2$ in the following way:
$$D_q=\frac{H_q(v)}{\log d}, \quad H_q(v)=\frac{1}{1-q}\log \left(\sum_{i=1}^d p_i^q\right).$$

Let us generalize this definition to an infinite-dimensional Hilbert space.
The measure $\mu_\lambda(\xi)$, which we introduced above, must carry most of its mass when $|\xi|^2$, the eigenvalue of the plane wave, is close to $\lambda$. Geometrically this means that the measure $\mu_\lambda$ will concentrate most of its mass inside an annulus of central radius $\sqrt{\lambda}$. 
 
Let us consider the case of \u{S}eba's billiard (cf. section \ref{sec-Seba} for an introduction), a rectangular billiard with a Dirac mass placed in its interior. The spectrum consists of two parts: old Laplace eigenvalues whose eigenfunctions do not feel the presence of the scatterer, and new eigenvalues whose eigenfunctions feel its presence. We will study the multifractal properties of these new eigenfunctions.

Let us denote the set of old Laplace eigenvalues (counted without multiplicity) by
$$\cN:=\{0=n_0<n_1<\dots<n_j<\cdots \to+\infty\}.$$
The set $\Lambda=\{\lambda_j\}_{j\geq 0}$ of new eigenvalues forms an interlacing sequence with $(n_j)_{j\geq 0}$.
We denote by $\Delta_j=\min_{m\in\cN}|m-\lambda_j|=|\tilde{n}_j-\lambda_j|$ the distance of a new eigenvalue $\lambda_j$ to the closest neighbouring Laplace eigenvalue, which we denote by $\tilde{n}_j$. If two such eigenvalues exist we denote by $\tilde{n}_j$ the smaller of the two.

Moreover, let us introduce the mean distance $$\ave{\Delta_j}_x=\frac{1}{\#\{k \mid \lambda_k\leq x\}}\sum_{\lambda_k\leq x}\Delta_k.$$

We distinguish two quantizations of \u{S}eba's billiard: 
\begin{itemize}
\item[(i)]{
weak coupling: $\ave{\Delta_j}_x=O((\log x)^{-1/2})$}
\item[(ii)]{
strong coupling: $\ave{\Delta_j}_x= (\log x)^{\alpha+o(1)}$, \;$\alpha=\alpha(\Lambda)\in(-\tfrac{1}{2},\tfrac{1}{2}]$.} 
\end{itemize}
Whereas the eigenvalues of the weak coupling quantization arise from self-adjoint extension theory, a renormalization of the extension parameter gives rise to the strong coupling quantization condition. For any given $\alpha\in(-\tfrac{1}{2},\tfrac{1}{2}]$, there exists a suitable renormalization to render the corresponding quantization. We refer the reader to section \ref{sec-Seba} for a detailed discussion. 

\begin{remark}
Note, if $\alpha(\Lambda)<\tfrac{1}{2}$, this does not preclude the mean spacing $\ave{\delta_j}_x$, where $\delta_j=n_j-\lambda_j$, from being of the same order as the mean spacing of the Laplace eigenvalues: $\ave{\delta_j}_x\asymp \sqrt{\log x}$\footnote{Note that the mean spacing of the Laplace eigenvalues $n_j$ is of order $\sqrt{\log x}$ as a consequence of Landau's Theorem on the representation of integers as a sum of two squares which states that $\#\{n_j\leq X\}\sim Bx/\sqrt{\log x}$.}. It simply means that the locations of the new eigenvalues fluctuate rather than remaining close to the intermediate values $\bar{n}_j=\tfrac{1}{2}(n_{j+1}+n_j)$.
\end{remark}

Let us return to the moment sum \eqref{moment_sum}. In the case of \u{S}eba's billiard we have the following explicit formula for the probability measure $\mu_\lambda$ on the lattice $\Z^2$:
$$\mu_\lambda(\xi)=(|\xi|^2-\lambda)^{-2}\left(\sum_{\xi\in\Z^2}(|\xi|^2-\lambda)^{-2}\right)^{-1}$$
which allows for a precise analysis.

Firstly, let us determine the scaling parameter. Since we are dealing with an infinite-dimensional space, the dimension must be replaced with the number of important components $N$ in the sum \eqref{moment_sum}. $N$ will be defined in terms of the average number of lattice points in the annulus which supports most of the mass of the measure $\mu_\lambda$; in particular $N$ will depend on $\lambda$.

In the case of \u{S}eba's billiard, the moment sum \eqref{moment_sum} is of the form
\begin{equation}\label{moment-sum}
M_q(\lambda_j)
=\frac{\zeta_{\lambda_j}(2q)}{\zeta_{\lambda_j}(2)^q}
\end{equation}
where we introduce the zeta function
\begin{equation}\label{zeta}
\zeta_\lambda(s)=\sum_{m\in\cN}r_Q(m)|m-\lambda|^{-s},\quad \re s>1,
\end{equation}
and $r_Q(n)$ denotes the representation number of the quadratic form defined by $Q(\xi_1,\xi_2)=|\xi|^2$ for $\xi=(\xi_1,\xi_2)\in\cL$, whereas $\cN$ denotes the set of numbers which is representable as a value of $Q$ for some $\xi\in\cL$. 
In the following we will consider the {\em arithmetic case}, when $\cL=\Z^2$ and $r_Q(n)=r_2(n)$.

\subsection{A first example: weak coupling}
In the case of the weak coupling quantization we have $\ave{\Delta_j}_n=O(1/\sqrt{\log n})$. In section \ref{simple} we will show that for a generic new eigenvalue only the lattice points on a circle of radius $\sqrt{n}$ contribute, i.e. the number of important terms in the sum \eqref{moment_sum} is $r_2(n)$. 
 
In fact, the moment sum can be rewritten as
$$M_q(\lambda_j)=\frac{\zeta_{\lambda_j}(2q)}{\zeta_{\lambda_j}(2)^q}=\frac{m_q(\lambda_j)}{m_1(\lambda_j)^q}$$
where 
\begin{equation}
m_q(\lambda_j)=r_2(\tilde{n}_j)+\Delta_j^{2q}\sum_{m\neq \tilde{n}_j}r_2(m)|m-\lambda_j|^{-2q}
\end{equation}
and, because $\Delta_j$ is typically small, only the term $r_2(\tilde{n}_j)$ contributes in the numerator and denominator respectively. In particular, we will show that, along a full density subsequence of eigenvalues, $$M_q(\lambda_j)\sim r_2(\tilde{n}_j)^{1-q}.$$

For a full density subsequence of $n\in\cN$ we have $$r_2(n)=(\log n)^{\tfrac{1}{2}\log 2+o(1)}$$ because $r_2(n)$ is close to its logarithmic normal order. Therefore, we define the scaling parameter as the logarithmic normal order of $r_2(\tilde{n}_j)$: 
$$N=(\log \tilde{n}_j)^{\tfrac{1}{2}\log 2}$$

We define the semiclassical fractal exponent, associated with a sequence of eigenvalues, as follows. 
\begin{defn}
Let $q>\tfrac{1}{2}$. Let $\Lambda$ be a sequence of eigenvalues which accumulates at infinity. 
We define the {\bf general entropy} associated with the eigenfunction $\psi_\lambda$ as
\begin{equation}
h_q(\psi_\lambda)=\log m_q(\lambda).
\end{equation}
 
For $q\neq1$ the {\bf Renyi entropy} associated with the probability measure $\mu_\lambda$ on the lattice $\cL$ may then be defined as
\begin{equation}
H_q(\mu_\lambda)=\frac{h_q(\psi_\lambda)-q h_1(\psi_\lambda)}{1-q}=\frac{1}{1-q}\log \left(\sum_{\xi\in\cL}\mu_\lambda(\xi)^q\right).
\end{equation}
The fractal exponents associated with $\Lambda$ are defined as
\begin{align}
d_q^\Lambda=&\limsup_{\lambda\in\Lambda \to+\infty}\frac{h_q(\psi_\lambda)}{\log N} \quad \text{with respect to the $\ell^\infty$-normalization,}\\
D_q^\Lambda=&\limsup_{\lambda\in\Lambda \to+\infty}\frac{H_q(\mu_\lambda)}{\log N} \quad \text{with respect to the $\ell^2$-normalization.}
\end{align}
\end{defn} 
\begin{remark}
We point out that the Renyi entropy of a probability measure on a lattice $\cL$ is a natural generalisation of the Shannon entropy familiar from information theory \cite{Sh48a,Sh48b}. For a probability measure $p=p_\xi$ on a lattice $\cL$, the Shannon entropy is defined as follows
$$H_{Sh.}(p)=-\sum_{\xi\in\cL}p_\xi \log p_\xi.$$

In fact, we have
\begin{align*}
\lim_{q\to 1} H_q(p)&=\lim_{q\to 1}\frac{1}{1-q}\log\left(\sum_{\xi\in\cL}p_\xi^q\right)\\
&=-\left[\frac{d}{dq}\log \left(\sum_{\xi\in\cL}p_\xi^q\right)\right]_{q=1}\\
&=-\sum_{\xi\in\cL}p_\xi \log p_\xi.
\end{align*}

So the Shannon entropy of $p$ is recovered from the Renyi entropy in the limit, as $q\to1$.
\end{remark}

We have the following theorem which we prove in section \ref{simple}. Due to the weakness of the perturbation, the new eigenvalue is on average close to a neighbouring Laplace eigenvalue. As a result we only have a single fractal exponent.
\begin{thm}[A simple scaling law]\label{simple_thm}
Let $\Lambda_w$ denote the sequence of new eigenvalues in the weak coupling regime. There exists a full density subsequence $\Lambda'\subset\Lambda_w$ such that $d_q(\Lambda')=D_q(\Lambda')=1$.
\end{thm} 

\subsection{A multifractal regime: strong coupling}
To determine the scaling parameter $N$ we have to calculate the number of important terms in the moment sum \eqref{moment_sum}. This is closely related to the number of lattice points in an annulus $$A(\lambda,G)=\{v\in\R^2 \mid ||\xi|^2-\lambda|\leq G\},$$ where $G=G(\lambda)$ is just large enough such that the terms corresponding to lattice points outside the annulus $A(\lambda,G)$ are negligible, when $\lambda\to+\infty$, and $G=o(\lambda)$. In particular, $G$ must grow faster than the mean spacing of the Laplace eigenvalues: $G/\sqrt{\log\lambda}\to +\infty$
 
Let us introduce the tail of the zeta function $\zeta_\lambda$:
$$\tau_q(\lambda,G)=\sum_{\substack{m\in\cN \\ |m-\lambda|\geq G}}r_2(m)|m-\lambda|^{-2q}.$$

\begin{defn}
We define the {\bf essential support} of the measure $\mu_\lambda$ as the set of lattice points $A(\lambda,G)\cap\Z^2$,
where $G=G(\lambda)$ satisfies $G/\sqrt{\log\lambda}\to +\infty$, and is determined from the condition
\begin{equation}\label{tail_neg}
\ave{\Delta_j}_\lambda^{2q}\ave{\tau_q}_\lambda=1
\end{equation}
and the mean tail is defined as
\begin{equation}
\ave{\tau_q}_T=\frac{1}{T}\int_T^{2T}\tau_q(t,G)dt.
\end{equation}
\end{defn}
\begin{remark}
We will see below that $\ave{\tau_q} \asymp G^{1-2q}$. So if we recall $\ave{\Delta_j}_\lambda=(\log \lambda)^{\alpha+o(1)}$, then we see that the definition of $G$ as a solution of \eqref{tail_neg} is quite robust when $\lambda\to+\infty$. Strictly speaking, for the terms outside the annulus $A(\lambda,G)$ to be negligible, one should require the r.h.s. of equation \eqref{tail_neg} to be $o(1)$. However, for any function which decays sufficiently slowly ($1/o(\log \lambda)$ to be precise) this will only change a lower order term in the asymptotic of $\log G$. Because, as we will see below, it is only the leading term in the asymptotic of $\log G$ which is relevant in the definition of $D_q$, we set the r.h.s. of \eqref{tail_neg} equal to $1$ for simplicity.
\end{remark}

As the number of lattice points in a thin annulus fluctuates a lot, we define the scaling parameter $N$ as the average number of lattice points in the annulus $A(\lambda,G)$ which is given by its area.
\begin{defn}
We define the {\bf scaling parameter} as $N=\Vol(A(\lambda,G))=2\pi G$.
\end{defn} 
 
We have the following proposition which we prove in section \ref{sec-scaling}.
\begin{prop}\label{scaling}
Suppose $q>1/2$. We have $$\ave{\tau_q}_\lambda \sim \frac{2\pi}{2q-1}G^{1-2q}$$
as $\lambda\to\infty$.


\end{prop}


Our main result concerns the strong coupling regime. In this regime the spacing of the new eigenvalues may be as large as the mean spacing of the Laplace eigenvalues. We determine under which conditions on the average distance of the new eigenvalues to the nearest Laplace eigenvalue a spectrum of fractal exponents emerges. Our result includes the case of a quantization consistent with level repulsion, provided the fluctuations around the mean location of the new eigenvalues are large enough. We prove the following theorem in section \ref{sec-multi}.
\begin{thm}[A multifractal scaling law]\label{multi_thm}
Let $\Lambda_s$ denote a sequence of new eigenvalues such that $\alpha=\alpha(\Lambda_s)\in(\tfrac{1}{4},\tfrac{1}{2})$. There exists a full density subsequence $\Lambda'\subset\Lambda_s$ such that for any $q>1$ which satisfies
 $(1-\log 2)(2-4\alpha)^{-1}< q \leq (2-4\alpha)^{-1}$
 we have the following formulae, which may be analytically continued to the complex plane,
\begin{equation}
d_q(\Lambda')=\frac{1}{2\alpha}\left(1-\frac{1}{2q}\right)\log2
\end{equation}
and for a constant $c\in[\tfrac{1}{2}\log 2,1]$
\begin{equation}
D_q(\Lambda')=\frac{1}{2\alpha}\left(1-\frac{1}{2q}\right)\frac{2cq-\log 2}{q-1}.
\end{equation}
\end{thm}
\begin{remark}
It is important to remark that ``rigid'' quantizations which have the property that the new eigenvalues stay close to the intermediate values of the Laplace eigenvalues, in the sense that $\ave{\delta_j}_x\asymp \ave{\Delta_j}_x =(\log x)^{1/2+o(1)}$, may lead to a possible breakdown of multifractality. This is interesting, because it suggests that in a regime, where level repulsion occurs, we also require the fluctuations of the new eigenvalues around their mean location to be large enough for multifractality to manifest itself. We discuss this case in detail in section \ref{sec-rigid}.
\end{remark}

\subsection{Multifractality of the ground state}\label{ground}
Instead of considering the semiclassical regime, where $\lambda\to+\infty$, the regime where $\lambda\to 0$ is also frequently studied. Physically this corresponds to the ground state in a weak coupling regime. In this regime we expect multifractality of the fluctuations of the eigenfunction around the constant term in the Fourier series \eqref{Fourier}, as $\lambda\to0$. This leads to the study of the modified moment sums
\begin{equation}
M_q^*(\lambda)=\zeta^*_\lambda(2q)
\end{equation}
where 
\begin{equation}\label{mod_zeta}
\zeta_\lambda^*(s)=\sum_{n\in\cN\setminus\{0\}}r_Q(n)|n-\lambda|^{-s},\quad \re s>1.
\end{equation}

The parameter $N$ does not depend on $\lambda$ in this regime and, therefore, we remove it from the definition of the fractal exponent.
\begin{defn}
Let $q>1/2$. We define the fractal exponent as follows
$$d_q^*=\lim_{\lambda\to 0}\log \zeta_\lambda^*(2q)$$
and
$$D_q^*=\frac{d_q^*-q d_1^*}{1-q}$$ 
\end{defn}

We state the second main result of this paper which provides an explicit formula for $D_q^*$ in terms of Epstein's zeta function
$$\zeta_Q(s)=\sum_{(m,n)\neq(0,0)}Q(m,n)^{-s}, \quad \re s>1$$
which is known to satisfy the functional equation $$\zeta_Q(s)=\varphi_Q(s)\zeta_Q(1-s).$$

Moreover, we show that $D_q^*$ satisfies a symmetry relation with respect to $q=\tfrac{1}{4}$ which is closely related to the functional equation of Epstein's zeta function. The following theorem is proved in section \ref{sec-ground}.
\begin{thm}\label{ground_frac} 
The fractal exponent $D^*_q$ admits an analytic continuation to the full complex plane in $q$. It satisfies the following symmetry relation with respect to the parameter value $q=\tfrac{1}{4}$:
$$D_{1/2-q}^*=\frac{1-q}{\tfrac{1}{2}+q}\left(D_q^* +\frac{\log\varphi_Q(2q)+(2q-\tfrac{1}{2})\log \zeta_Q(2)}{1-q}\right).$$

In particular, in the limit as $q\to 1$, the fractal exponent converges to the Shannon entropy of the measure $\mu_\lambda(\xi)=\zeta_Q(2)^{-1}|\xi|^{-4}$: $$\lim_{q\to 1}D_q^*=\log \zeta_Q(2)-2\frac{\zeta_Q'(2)}{\zeta_Q(2)}$$
which we may rewrite as 
$$H_{Sh.}(\mu_\lambda)=-\sum_{\xi\in\Z^2}\frac{1}{\zeta_Q(2)|\xi|^4}\log\left(\frac{1}{\zeta_Q(2)|\xi|^4}\right).$$ 
\end{thm}

\section{Background on \u{S}eba billiards}\label{sec-Seba}

\subsection{Spectral theory of \u{S}eba billiards}
\u{S}eba billiards belong to the class of pseudointegrable billiards and over the past 30 years they have been studied extensively in the literature. For a detailed discussion of the spectral theory of such billiards we refer the reader to the introductory sections of the articles \cite{RU12} and \cite{U14}.

In this article we will deal with an arithmetic \u{S}eba billiard associated with the square lattice $\Z^2$. 
Let $\T^2=\R^2/2\pi\Z^2$ and $x_0\in\T^2$. The spectrum of the positive Laplacian $-\Delta=-(\partial_x^2+\partial_y^2)$ is given by integers which may be represented as a sum of two squares $n=x^2+y^2$, $(x,y)\in\Z^2$. We denote the ordered set of such numbers as $\cN$. The multiplicities in the Laplace spectrum on $\T^2$ are, thus, given by the representation number $$r_2(n)=\#\{(x,y)\in\Z^2 \mid n=x^2+y^2\}$$ which is an arithmetic function that has been extensively studied in the number theory literature.

Consider a self-adjoint extension $-\Delta_\varphi$, $\varphi\in[-\pi,\pi)$, of  $-\Delta$ restricted to the domain $C_c^\infty(\T^2\setminus\{x_0\})$. Physically this operator corresponds to the {\bf weak coupling quantization} for a square torus with a Dirac delta potential (as introduced by \u{S}eba in \cite{Se90}).

The spectrum of the operator $-\Delta_\varphi$ consists of two parts: 
\begin{itemize}
\item[(i)]{old (Laplace) eigenvalues with multiplicity reduced by $1$}
\item[(ii)]{new eigenvalues with multiplicity $1$ which interlace with the Laplace eigenvalues}
\end{itemize}

The eigenspace associated with an old eigenvalues is simply the codimension one subspace of Laplace eigenfunctions which vanish at $x_0$. The old eigenfunctions do not feel the presence of the potential. In this article, we will only be interested in the new eigenfunctions which {\em do} feel the presence of the potential. 

These new eigenfunctions are Green's functions associated with the resolvent of the Laplacian, where one variable has been fixed as the position of the potential $x_0$. Because of translation invariance we may set $x_0=0$.
We then have the following $L^2$-representation of the new eigenfunctions
\begin{equation}
\begin{split}
G_\lambda(x)\stackrel{L^2}{=}&\sum_{\xi\in\Z^2}\frac{e^{i\left\langle\xi,x\right\rangle}}{|\xi|^2-\lambda}\\
=&\sum_{n\in\cN}\frac{1}{n-\lambda}\sum_{|\xi|^2=n}e^{i\left\langle\xi,x\right\rangle}.
\end{split}
\end{equation}

Moreover, the new eigenvalues may be determined as the solutions of the following spectral equation
\begin{equation}\label{spec-eq}
\sum_{n\in\cN}r_2(n)\left\{\frac{1}{n-\lambda}-\frac{n}{n^2+1}\right\}=c_0\tan(\frac{\varphi}{2}).
\end{equation} 

Because the function on the l.h.s. of equation \eqref{spec-eq} is monotonous on intervals $(n_j,n_{j+1})$, $n_j\in\cN$, we easily see that the new eigenvalues form a sequence of real numbers which interlace with the sequence of Laplace eigenvalues:
$$\lambda_0<0=n_0<\lambda_1<n_1<\lambda_2<n_2<\cdots<\lambda_j<n_j<\cdots \to +\infty.$$

An important question concerns the statistics of the spacings $\delta_j:=n_j-\lambda_j$. As was shown in \cite{RU14}, $\lambda_j$ is on average close to $n_j$, precisely $\delta_j=O((\log n_j)^{-1/2})$ for a full density subsequence of Laplace eigenvalues. This shows that the quantization procedure using self-adjoint extension theory, which was introduced by \u{S}eba \cite{Se90}, gives rise to a so-called weak coupling quantization. Shigehara \cite{Sh94} and later Bogomolny, Gerland and Schmit \cite{BoGeSch01} explained that a semiclassical renormalization of the extension parameter $\varphi$ was necessary in order to obtain a strong coupling quantization giving rise to a spectrum and eigenfunctions which are expected to have interesting features.

\subsection{Renormalization and strong coupling quantizations}
In fact, a logarithmic renormalization of the coupling parameter, $\tan(\varphi_\lambda/2)\sim c\log\lambda$, is equivalent to a truncation of the sum on the l.h.s. of equation \eqref{spec-eq} outside a window of size $\lambda^{1/2}$. We refer the reader to \cite{U14} for a detailed discussion of this renormalization procedure. The reason for this equivalence is that the tail of the sum has a logarithmic asymptotic as $\lambda\to+\infty$ which is cancelled by the leading term in the asymptotic of $\tan(\varphi_\lambda/2)$. The second order term in this asymptotic is related to local information about the position of the new eigenvalues.

We thus obtain the following strong coupling quantization, where the function $\beta=\beta(\lambda)=o(\log\lambda)$ is related to the inverse coupling strength of the potential, and $\lambda_j$ is determined as the solution on the interval $(n_j,n_{j+1})$ as the solution to the following equation
\begin{equation}\label{spec-eq-s}
\sum_{\substack{n\in\cN \\ |n-n_j|\leq n_j^{1/2}}}\frac{r_2(n)}{n-\lambda}=\beta(\lambda).
\end{equation}

The relationship between the precise form of $\beta$ and the location of the new eigenvalues is not understood at the rigorous level. When $\beta$ is constant, we expect the spacings $\delta_j=n_j-\lambda_j$ to be of the order of the mean spacing of the Laplace eigenvalues, $\ave{\delta_j}_x\sim c\sqrt{\log x}$. However, one may choose $\beta=c(\log\lambda)^b$, $b\in(0,1)$, to achieve spacings of lower order.

It is important to remark that the location of a new eigenvalue $\lambda_j$ within the interval $(n_j,n_{j+1})$ depends on the interplay of the coefficients $r_2(n)$ in the summation in equation \eqref{spec-eq-s}. Due to the strong fluctuations of the function $r_2(n)$ one may have $\ave{\delta_j}_x\asymp \sqrt{\log x}$, even though $\ave{\Delta_j}_x\ll (\log x)^\alpha$, $\alpha<\tfrac{1}{2}$.

\section{An estimate for $\zeta_{\lambda_j}(2q)$}

In this section we will prove the following lemma. 
\begin{lem}\label{tail_bound_0}
Let $q>1$.
There exists a subsequence of full density $\cN_1\subset\cN$ such that for any $m\in\cN_1$ we have the bound
$$\sum_{\substack{n\in\cN \\ n\neq m}}\frac{r_2(n)}{|m-n|^{2q}}\ll (\log m)^{-q+1/2+o(1)}.$$
\end{lem}
\begin{proof}
Let $m\in\cN_1(x)$. For convenience denote $L=\log m$. We first of all claim that there is a full density subsequence of $\cN_1\subset\cN$ such that for $m\in\cN_1$ (proof to be given in a separate section)
\begin{itemize}
\item[(i)]
{if $r\in[L^{1/2+o(1)},L^2]$ there are 
$\gg r/((\log m)^{1/2+o(1)}(\log r)^2)$ elements in $\cN$ that lie in between $m$ and $n$ if $r_2(n)\geq r$}
\item[(ii)]
{if $r\geq L^2$ there are $\gg_\epsilon r^{5/4-\epsilon}/\log r$ elements in $\cN$ that lie in between $m$ and $n$}
\item[(iii)]
{the nearest neighbours of $m$ are at distance at least $L^{1/2-o(1)}$}
\item[(iv)]
{for $G\in[2,m^{1-\epsilon}]$ we have the bound 
$$\sum_{\substack{n\in\cN \\ |m-n|\geq G}}|m-n|^{-2q} \leq \frac{(\log G)^2}{G^{2q-1}L^{1/2-o(1)}}$$}
\end{itemize}

We next divide the sum in three parts according to the size of $r_2(n)$:

{\em Small $r_2(n)$:} $r_2(n)\leq L^{1/2+o(1)}$. We have
\begin{align*}
\sum_{\substack{|m-n|\gg L^{1/2-o(1)}\\ r_2(n)\leq L^{1/2+o(1)}}}\frac{r_2(n)}{|m-n|^{2q}}
\ll& \frac{L^{1/2+o(1)}}{L^{q-1/2-o(1)}L^{1/2-o(1)}}=L^{-q+1/2+o(1)}
\end{align*}

{\em Medium $r_2(n)$:}
Set $R_i=2^i(\log m)^{1/2+o(1)}$. By property (6) (cf. \cite{KU14}) we may set $T=|m-n|$ and we then obtain
$$\frac{T(\log T)^2}{(\log m)^{1/2-o(1)}}\gg \frac{R_i}{L^{1/2+o(1)}(\log R_i)^2}$$
which implies $T\gg R_i/L^{o(1)}$.
 
We thus have 
\begin{equation}
\begin{split}
\sum_{r_2(n)\in[R_i,R_{i+1}]}\frac{r_2(n)}{|m-n|^{2q}} 
&\ll \sum_{\substack{n\in\cN \\ |m-n|\gg R_i/L^{o(1)}}}\frac{R_i}{|m-n|^{2q}}\\
&\ll \frac{R_i(\log R_i)^2}{L^{1/2-o(1)}R_i^{2q-1}}\\
&=\frac{i^2}{L^{q-1/2-o(1)}2^{i(2q-2)}} 
\end{split}
\end{equation}
and summing over $i\geq0$ yields the result (we used $q>1)$.

{\em Large $r_2(n)$:} We imitate the argument used in the proof of property (8) in \cite{KU14}.\\
Take $R_i=2^i L^2/100$. Suppose $r_2(n)\in[R_i,2R_i]$. By property (ii) we have $|m-n|\gg_\epsilon R_i^{5/4-\epsilon}$.
And, thus, we have
$$\sum_{r_2(n)\in[R_i,2R_i]}\frac{r_2(n)}{|m-n|^{2q}}
\ll R_i\sum_{|m-n|\gg_\epsilon R_i^{5/4-\epsilon}}|m-n|^{-2q}
\ll_\epsilon R_i(R_i^{5/4-\epsilon})^{1-2q}L^{-1/2+o(1)}$$
and this is bounded by $2^{i(9/4-5q/2+\epsilon)}L^{4-5q+\epsilon}$. 

Summing over $i$ (recall $q>1$, so $9/4-5q/2<0$) we obtain
$$\sum_{r_2(n)\geq L^2}\frac{r_2(n)}{|m-n|^{2q}}=O_\epsilon(L^{4-5q+\epsilon})=O(L^{1-2q})$$
because $q>1$ implies $4-5q<1-2q$.

\end{proof}

\subsection{Proof that $\cN_1$ is of full density:}
\begin{itemize}
\item[(i)]{
We use the same strategy as in the proof of property (7) in \cite{KU14}. 
Given that $n\in\cN$ satisfies $r=r_2(n)\in[(\log n)^{1/2}\log\log n,(\log n)^2]$ we remove \\
$2r/((\log n)^{1/2}(\log\log n)(\log r)^2)$ elements of $\cN$ to the left and right of $n$ and denote the set of such removed elements by $R_n$. As was shown in the proof of property (7), it is sufficient to bound the total number of removed elements in all sets $R_n$ such that $n\leq 2x$. If we show that this set is of zero density in $\cN$ (i.e. its cardinality is $o(x/\sqrt{\log x})$), then we have shown that the set of remaining elements is of full density in $\cN$. Moreover, for any $m$ in this full density set and any $n\in\cN$ such that $r_2(n)\geq r$ and $r\in[(\log n)^{1/2}\log\log n,(\log n)^2]$ we have at least $r/((\log n)^{1/2}(\log\log n)(\log r)^2)$ elements of $\cN$ between $m$ and $n$.

It remains to be shown that the union of $R_n$, $n\leq 2x$, is indeed of cardinality $o(x/\sqrt{\log x})$. As in the proof of (7), it is sufficient (in view of the upper bound $(\log n)^2$) to consider $n\in[x/(\log x)^{10},2x]$ and we, thus, have $\log n=(1+o(1))\log x$.

In view of the asymptotic $$\sum_{n\leq x}r_2(n)\sim \pi x$$ Chebyshev gives that the number of $n$ such that $r_2(n)\geq T$ is $\ll x/T$.

Let $r_2(n)\in[R_i,2R_i]$ with $R_i=2^i(\log x)^{1/2}\log\log x$. So the total number of removed elements is
$$\ll \frac{x}{R_i}\frac{R_i}{(\log x)^{1/2}(\log\log x)(\log R_i)^2}\ll \frac{x}{(\log x)^{1/2}(\log\log x)i^2}$$
and summing over $i$ shows that the total number of elements removed is $o(x/\sqrt{\log x})$.
}\\

\item[(ii)]{
We follow the same strategy as in the proof of property (8) in \cite{KU14}. Given $n\in\cN$ such that $r=r_2(n)\geq(\log n)^2$ we remove $2r^{5/4-\epsilon}/\log r$ elements of $\cN$ to the left and right. We use the bound $$\sum_{n\leq x}r_2(n)^2\ll x\log x.$$ As in the proof of (8) it is sufficient to consider $n\in[\sqrt{x},2x]$. Let $R_i=2^i(\log x)^2/100$ and consider $r\in[R_i,2R_i]$. Chebyshev gives that the number of
$n\leq 2x$ s. t. $r_2(n)\geq R_i$ is $\ll \frac{x\log x}{R_i^2}$ and the associated removed contribution is
$$\ll \frac{x\log x}{R_i^2}\frac{R_i^{5/4-\epsilon}}{\log R_i}
\ll x\log x ~ R_i^{-3/4-\epsilon}=x2^{-i(3/4+\epsilon)}(\log x)^{-1/2-\epsilon}$$
and summing over $i$ shows that the total number of removed elements is $o(x/\sqrt{\log x})$.
}\\

\item[(iii)]{This is the same as property (3) in \cite{KU14}.}\\

\item[(iv)]{
The proof is analogous to the proof of property (9) in \cite{KU14}, where $(m-n)^{-2}$ is replaced with $|m-n|^{-2q}$ and we use
$$\sum_{h\geq T}\frac{c(h)}{h^{2q}}\ll T^{1-2q}.$$ 
}
\end{itemize}

\section{Proofs of Theorems \ref{multi_thm} and \ref{simple_thm}} 

Let us consider the moments
\begin{align*}
m_q(\lambda_j)=&\Delta_j^{2q}\sum_{\xi\in\Z^2} \frac{1}{||\xi|^2-\lambda|^{2q}}\\
=&(\tilde{n}_j-\lambda_j)^{2q}\sum_{n\in\cN}r_2(n)|n-\lambda_j|^{-2q}, \quad q>1,
\end{align*}
which give information about the multifractal properties of the eigenfunctions.


\subsection{Proof of Theorem \ref{multi_thm}}\label{sec-multi}
Let us denote
$$\cN_2=\{n\in\cN_1 \mid r_2(n)=(\log n)^{\tfrac{1}{2}\log 2+o(1)}\}$$
which is of full density in $\cN$ (cf. \cite{KU14}).

Moreover, $\alpha(\Lambda)\in(0,\tfrac{1}{2})$ and Chebyshev imply that $\Delta_j\ll (\log \tilde{n}_j)^{\alpha+o(1)}$ is a full density condition if we assume $(\log \tilde{n}_j)^{o(1)}\nearrow+\infty$. We denote by $\cN'$ the subset of $\cN_2$ which satifies this condition.

We need to estimate the fractal exponent associated with the subsequence $\cN'$
$$D_q=\frac{1}{1-q}\limsup_{\tilde{n}_j\in\cN'}\frac{\log m_q(\tilde{n}_j)-q\log m_1(\tilde{n}_j)}{\log N}.$$

Let $L=\log \tilde{n}_j$. For $q>1$ and $\tilde{n}_j\in\cN'$ we have
\begin{equation}
\begin{split}
m_q(\tilde{n}_j)=~ &r_2(\tilde{n}_j) + \Delta_j^{2q}O(L^{1/2-q+o(1)})\\
=~ &L^{\log 2/2+o(1)}+O(L^{2\alpha q+o(1)})\times O(L^{1/2-q+o(1)})
\end{split}
\end{equation}
where we used Lemma \ref{tail_bound_0}.

We have $q>(1-\log 2)/(2-4\alpha)$ which implies $-(1-2\alpha)q+\tfrac{1}{2}<\tfrac{1}{2}\log 2$
so that we have the asymptotic
$$m_q(\tilde{n}_j)\sim~ L^{\log 2/2+o(1)}.$$

Thus, we find
$$\limsup_{\tilde{n}_j\in\cN'}\frac{\log m_q(\tilde{n}_j)}{\log\log \tilde{n}_j}=\lim_{\tilde{n}_j\in\cN'}\frac{\log m_q(\tilde{n}_j)}{\log\log \tilde{n}_j}=\frac{1}{2}\log 2.$$

We need a separate estimate for $m_1(\tilde{n}_j)$, because the above estimate assumes $q$ to be sufficiently large. 

For $q=1$ we have
$$\sum_{n\in\cN, \; n\neq \tilde{n}_j}\frac{r_2(n)}{(n-\tilde{n}_j)^2}
=\sum_{\substack{n\in\cN, \; n\neq \tilde{n}_j\\ r_2(n)\leq L^2}}\frac{r_2(n)}{(n-\tilde{n}_j)^2} 
+\sum_{\substack{n\in\cN, \; n\neq \tilde{n}_j\\ r_2(n)> L^2}}\frac{r_2(n)}{(n-\tilde{n}_j)^2}$$
where the second sum is bounded by $O_\epsilon(L^{-1+\epsilon})$ in view of the bound obtained for sums over large $r_2(n)$ in the proof of Lemma \ref{tail_bound}, and the first sum is bounded by
$$L^2\sum_{|n-\tilde{n}_j|\gg L^{1/2-o(1)}}(n-\tilde{n}_j)^{-2}\ll L^{1+o(1)}$$ where we used (9) in \cite{KU14}.

This yields the bound
\begin{equation}\label{q1_bound}
\sum_{n\in\cN, \; n\neq \tilde{n}_j}\frac{r_2(n)}{(n-\tilde{n}_j)^2}\ll L^{1+o(1)}.
\end{equation}

So, we see that $$c:=\liminf_{n\in\cN'}\frac{\log m_1(n)}{\log\log n}\in[\tfrac{1}{2}\log 2,1].$$

As a consequence of Prop. \ref{scaling}, we see that
$$\log\ave{\tau_q}_n=(1-2q)\log G+O(\log q)+o(1)$$
and from the definition of $G$ we have
$$\log\ave{\tau_q}_n=-2q\log \ave{\Delta_j}_n=-2q\, \alpha\log\log n+o(\log\log n).$$

Combining these equations we obtain
\begin{equation}
\log G=\alpha(1+o(1))\frac{2q}{2q-1}\log\log n,
\end{equation}
and we recall the requirement $G/\sqrt{\log n}\to+\infty$. Therefore, we must have
$\alpha\frac{2q}{2q-1}<\frac{1}{2}$ which yields the upper bound $q<(2-4\alpha)^{-1}$.

Since $N=2\pi G$, we obtain
\begin{align*}
\lim_{\tilde{n}_j\in\cN'\to+\infty}\frac{h_q(\tilde{n}_j)}{\log N}
=&\frac{1}{\alpha}(1-\frac{1}{2q})\lim_{\tilde{n}_j\in\cN'\to+\infty}\frac{\log m_q(\tilde{n}_j)}{\log\log \tilde{n}_j} \\ 
=&\frac{1}{2\alpha}(1-\frac{1}{2q})\log 2
\end{align*}
and 
\begin{align*}
\liminf_{\tilde{n}_j\in\cN'\to+\infty}\frac{h_1(\tilde{n}_j)}{\log N}=&\frac{1}{\alpha}(1-\frac{1}{2q})\liminf_{n\in\cN'\to+\infty}\frac{\log m_1(\tilde{n}_j)}{\log\log \tilde{n}_j}\\
=&\frac{c}{\alpha}(1-\frac{1}{2q}).
\end{align*}

We, thus, obtain the following formula for the fractal exponent:
\begin{equation}
\begin{split}
D_q=&\limsup_{\tilde{n}_j\in\cN'\to+\infty}\frac{H_q(\tilde{n}_j)}{\log N}\\
=&\frac{1}{1-q}\left(\lim_{\tilde{n}_j\in\cN'\to+\infty}\frac{h_q(\tilde{n}_j)}{\log N}-q\liminf_{\tilde{n}_j\in\cN'\to+\infty}\frac{h_1(\tilde{n}_j)}{\log N}\right)\\
=&\frac{1}{2\alpha}\frac{2cq-\log 2}{q-1}\left(1-\frac{1}{2q}\right)
\end{split}
\end{equation}

\subsection{Proof of Theorem \ref{simple_thm}}\label{simple}

We know (cf. \cite{RU14}) that the new eigenvalues of the operator $-\Delta_\varphi$ are on average close to a nearby old Laplace eigenvalue. More precisely, there exists a full density subsequence $\cN_2\subset\cN$ such that $\Delta_m=O(1/\sqrt{\log m})$.

We have
$$\sum_{n\in\cN}\frac{r_2(n)}{|n-\lambda_m|^{2q}}
=\frac{r_2(\tilde{n}_j)}{|\tilde{n}_j-\lambda_j|^{2q}}+\sum_{\substack{n\in\cN \\ m\neq \tilde{n}_j}}\frac{r_2(m)}{|m-\lambda_j|^{2q}}$$
and, by applying Lemma \ref{tail_bound_0}, we get for $\tilde{n}_j\in\cN_1$ and $q>1$
$$\sum_{\substack{m\in\cN \\ m\neq \tilde{n}_j}}\frac{r_2(m)}{|m-\lambda_j|^{2q}}
=O((\log \tilde{n}_j)^{-q+1/2+o(1)}).$$

If $\tilde{n}_j\in\cN_1\cap\cN_2$, then, for any $q>1$,
$$m_q(\lambda_j)=(\tilde{n}_j-\lambda_j)^{2q}\sum_{m\in\cN}\frac{r_2(m)}{|m-\lambda_j|^{2q}}
= r_2(\tilde{n}_j)+O((\log \tilde{n}_j)^{-2q+1/2+o(1)})
$$
and thus the general entropy has asymptotics $$h_q(\tilde{n}_j)\sim \log r_2(\tilde{n}_j)$$
and $$d_q=\lim_{\tilde{n}_j\in \cN_1\to+\infty} \frac{h_q(\tilde{n}_j)}{\tfrac{1}{2}\log2\cdot\log\log \tilde{n}_j}=1$$
because $r_2(\tilde{n}_j)=(\log \tilde{n}_j)^{\tfrac{1}{2}\log 2+o(1)}$.

Moreover, for $q=1$, we recall the bound \eqref{q1_bound} to see that
\begin{equation}
m_1(\lambda_j)=r_2(\tilde{n}_j)+O(L^{o(1)})\sim (\log\tilde{n}_j)^{\tfrac{1}{2}\log 2+o(1)}
\end{equation}
where we used $\Delta_j^2=O((\log \tilde{n}_j)^{-1})$. It follows that
$$d_1=\lim_{\tilde{n}_j\in \cN_1\to+\infty} \frac{h_1(\tilde{n}_j)}{\tfrac{1}{2}\log2\cdot\log\log \tilde{n}_j}=1.$$

Hence, for $q\neq 1$, $$D_q=\frac{d_q-q d_1}{1-q}=1.$$

\section{Rigid quantizations}\label{sec-rigid}
 
In this section we will deal with quantizations which have the property that the new eigenvalues stay close to the intermediate values of the Laplace eigenvalues, $\bar{n}_j=\tfrac{1}{2}(n_j+n_{j+1})$, in the sense that $\ave{\Delta_j}_x \asymp \ave{\delta_j}_x = (\log x)^{1/2+o(1)}$.

Under the assumption of an analogue of the Sathe-Selberg asymptotic on $\omega_1(n)$, the number of distinct prime factors congruent to $1$ mod 4, we obtain the following lemma which improves the exponent in Lemma \ref{tail_bound_0} if $q$ is sufficiently large.
\begin{lem}\label{tail_bound}
Let $q>\tfrac{3}{2}$. There exists a subsequence of full density $\cN_3\subset\cN$ such that for any $m\in\cN_1$ we have the bound
$$\sum_{\substack{n\in\cN \\ n\neq m}}\frac{r_2(n)}{|m-n|^{2q}}\ll (\log m)^{-q+\tfrac{1}{2}\log2+f(q)+o(1)}$$ as $m\to+\infty$,
where $$f(q)=\frac{1}{2}\log 2\left(\exp\left\{\frac{\log 2}{q-1/2}\right\}-1\right).$$
\end{lem}
\begin{proof}
Let $m\in\cN_3(x)$. For convenience denote $L=\log m$. We first of all claim that there is a full density subsequence of $\cN_3\subset\cN$ such that for $m\in\cN_3$ (proof to be given in a separate section) the following properties hold. Note that properties (ii)-(iv) coincide with the respective properties of $\cN_1$, which we restate here for completeness.
\begin{itemize}
\item[(i)]
{if $r\in[L^{\log 2/2+\eta},L^2]$, and $r\in[2^iL^{\log 2/2+\eta},2^{i+1}L^{\log 2/2+\eta}]$ for some integer $i$, then there are 
$\gg (1+2\eta')^i/((\log r)^2 \log L )$ elements in $\cN$ that lie in between $m$ and $n$ if $2^{\omega_1(n)}\geq r$}
\item[(ii)]
{if $r\geq L^2$ there are $\gg_\epsilon r^{5/4-\epsilon}/\log r$ elements in $\cN$ that lie in between $m$ and $n$}
\item[(iii)]
{the nearest neighbours of $m$ are at distance at least $L^{1/2-o(1)}$}
\item[(iv)]
{for $G\in[2,m^{1-\epsilon}]$ we have the bound 
$$\sum_{\substack{n\in\cN \\ |m-n|\geq G}}|m-n|^{-2q} \leq \frac{(\log G)^2}{G^{2q-1}L^{1/2-o(1)}}$$}
\item[(v)]
{for $f(n):=r_2(n)/{4\cdot 2^{\omega_1(n)}}$ and $Q\geq1$ there are at least 
$Q/((\log Q)^2\log L)$ elements in $\cN$ that lie in between $m$ and $n$ if $f(n)\geq Q$
}
\end{itemize}

We next divide the sum in three parts according to the size of $r_2(n)$:

{\em Small $r_2(n)$:} $r_2(n)\leq L^{\log2/2+\epsilon}$. We have
\begin{align*}
\sum_{\substack{|m-n|\gg L^{1/2-o(1)} \\ r_2(n)\leq L^{\log2/2+\epsilon}}}\frac{r_2(n)}{(m-n)^{2q}}
\ll& \sum_{|m-n|\gg L^{1/2-o(1)}}\frac{r_2(n)}{(m-n)^{2q}}\\
\ll& \frac{L^{\log2/2+\epsilon}}{L^{q-1/2-o(1)}L^{1/2-o(1)}}=L^{-q+\log2/2+\epsilon}
\end{align*}
 
{\em Medium $r_2(n)\in[L^{\log 2/2+\eta},L^2]$:}
Set $R_i=2^i(\log m)^{\log 2/2+\eta}$. By property (6) (cf. \cite{KU14}) we may set $T=|m-n|$ and we then obtain, in view of (i),
$$\frac{T(\log T)^2}{(\log m)^{1/2-o(1)}}\gg \frac{(1+2\eta')^i}{(\log R_i)^2 \log L},$$ where $\eta'=\eta/\log 2$,
which implies $T\gg (1+2\eta')^i L^{1/2-o(1)}$.

Similarly, in view of (v), if $f(n)=r_2(n)/(4\cdot 2^{\omega_1(n)})\in[2^j,2^{j+1}]$, then we have
$$\frac{T(\log T)^2}{(\log m)^{1/2-o(1)}}\gg \frac{2^j}{j^2 \log L}$$
which implies $T\gg 2^j L^{1/2-o(1)}$.

Combining these two lower bounds we get $$|m-n|\gg G_{i,j}:=2^{j/2}(1+2\eta')^{i/2}L^{1/2-o(1)}.$$
 
We define $$\cN_{i,j}:=\{n\in\cN \mid f(n)\in[2^j,2^{j+1}] \; \text{and} \; 2^{\omega_1(n)}\in[R_i,R_{i+1}]\}.$$

Therefore,
$$\sum_{\substack{n\in\cN, \; n\neq m \\ r_2(n)\in[L^{\log 2/2+\eta},L^2]}}
\frac{r_2(n)}{|m-n|^{2q}}=\sum_{i,j\geq 0} \sum_{n\in\cN_{i,j}}\frac{r_2(n)}{|m-n|^{2q}}
\ll \sum_{i,j\geq0} 2^j \sum_{n\in\cN_{i,j}}\frac{2^{\omega_1(n)}}{|m-n|^{2q}}$$

Let us bound the interior sum:
\begin{equation}
\begin{split}
\sum_{n\in\cN_{i,j}}\frac{2^{\omega_1(n)}}{|m-n|^{2q}} 
&\ll \sum_{|m-n|\gg G_{i,j}}\frac{R_i}{|m-n|^{2q}}\\
&\ll \frac{R_i(\log G_{i,j})^2}{L^{q-o(1)}(1+2\eta')^{i(q-1/2)}2^{j(q-1/2)}}\\
&=\frac{2^i}{(1+2\eta')^{i(q-1/2)}2^{j(q-1/2)})} L^{-q+\log2/2+\eta+o(1)}
\end{split}
\end{equation}

And if $q>\max(\tfrac{3}{2},\tfrac{1}{2}+\frac{\log 2}{\log(1+2\eta')})$, we may sum over $i$ and $j$ to obtain 
$$\sum_{\substack{n\in\cN, \; n\neq m \\ r_2(n)\in[L^{\log 2/2+\eta},L^2]}}\frac{r_2(n)}{|m-n|^{2q}}\ll_\epsilon L^{-q+\log2/2+\eta+o(1)}.$$

So, if $q>\tfrac{3}{2}$, then we need
$$q>\tfrac{1}{2}+\frac{\log 2}{\log(1+2\eta/\log 2)}$$
which is achieved by $$\eta>\tfrac{1}{2}\log 2\left[\exp\left\{\frac{\log 2}{q-\tfrac{1}{2}}\right\}-1\right].$$

{\em Large $r_2(n)$:} We imitate the argument used in the proof of property (8) in \cite{KU14}.\\
Take $R_i=2^i L^2/100$. Suppose $r_2(n)\in[R_i,2R_i]$. By property (ii) we have $|m-n|\gg_\epsilon R_i^{5/4-\epsilon}$.
And, thus, we have
$$\sum_{r_2(n)\in[R_i,2R_i]}\frac{r_2(n)}{|m-n|^{2q}}
\ll R_i\sum_{|m-n|\gg_\epsilon R_i^{5/4-\epsilon}}|m-n|^{2q}
\ll_\epsilon R_i(R_i^{5/4-\epsilon})^{1-2q}L^{-1/2+o(1)}$$
and this is bounded by $2^{i(9/4-5q/2+\epsilon)}L^{4-5q+\epsilon}$. 

Summing over $i$ (recall $q>1$, so $9/4-5q/2<0$) we obtain
$$\sum_{r_2(n)\geq L^2}\frac{r_2(n)}{|m-n|^{2q}}=O_\epsilon(L^{4-5q+\epsilon})=O(L^{1-2q})$$
because $q>1$ implies $4-5q<1-2q$.

\end{proof}

\subsection{Proof that $\cN_3$ is of full density:}
\begin{itemize}
\item[(i)]{
We use the same strategy as in the proof of property (7) in \cite{KU14}. 
Given that $n\in\cN$ satisfies $r=2^{\omega_1(n)}\in[(\log n)^{\tfrac{1}{2}\log2+\eta},(\log n)^2]$
we have $r\in[2^i(\log n)^{\tfrac{1}{2}\log2+\eta},2^{i+1}(\log n)^{\tfrac{1}{2}\log2+\eta}]$ for some integer $i$.
We remove \\
$(1+2\eta')^i/((\log\log n)(\log r)^2)$ elements of $\cN$ to the left and right of $n$ and denote the set of such removed elements by $R_n$. As was shown in the proof of property (7), it is sufficient to bound the total number of removed elements in all sets $R_n$ such that $n\leq 2x$. If we show that this set is of zero density in $\cN$ (i.e. its cardinality is $o(x/\sqrt{\log x})$), then we have shown that the set of remaining elements is of full density in $\cN$. Moreover, for any $m$ in this full density set and any $n\in\cN$ such that $r_2(n)\geq r$ and $r\in[(\log n)^{\tfrac{1}{2}\log 2+\eta},(\log n)^2]$ with $r\in[2^i(\log n)^{\tfrac{1}{2}\log2+\eta},2^{i+1}(\log n)^{\tfrac{1}{2}\log2+\eta}]$, we have at least $(1+2\eta')^i/((\log\log n)(\log r)^2)$ elements of $\cN$ between $m$ and $n$.

It remains to be shown that the union of $R_n$, $n\leq 2x$, is indeed of cardinality $o(x/\sqrt{\log x})$. As in the proof of (7), it is sufficient (in view of the upper bound $(\log n)^2$) to consider $n\in[x/(\log x)^{10},2x]$ and we, thus, have $\log n=(1+o(1))\log x$.

Recall that $\omega_1(n)$ denotes the number of distinct prime factors congruent to $1$ mod 4. The analogue of the Sathe-Selberg asymptotic for $\omega_1(n)$ gives 
\begin{equation}\label{delange}
\#\{n\in\cN(x) \mid \omega_1(n)=k\} \sim C_0 \frac{x}{\log x}\frac{(\log\log x)^{k-1}}{2^{k-1}(k-1)!}
\end{equation}
for $k\leq C\log\log x$, where $C$ is a constant.

Let $r=2^{\omega_1(n)}\in[R_i,R_{i+1}]$ with $R_i=2^i(\log x)^{\tfrac{1}{2}\log 2+\eta}$ 


Because $\omega_1(n)=\log r/\log 2$, we see that $r\in[R_i,R_{i+1}]$ is equivalent to
$$\omega_1(n)\in [(\tfrac{1}{2}+\eta')\log\log x +i,(\tfrac{1}{2}+\eta')\log\log x+i+1]$$
where $\eta'=\eta/\log 2$.
Thus, the asymptotic \eqref{delange} yields that the number of $n\in\cN$ satisfying $r\in[R_i,R_{i+1}]$ is bounded by
$$\ll \left(\frac{1}{1+2\eta'}\right)^i\frac{x}{\sqrt{\log x}\sqrt{\log\log x}}\ll \frac{x}{(1+2\eta')^i\sqrt{\log x}}.$$ 

So the total number of removed elements is
$$\ll \frac{x}{(1+2\eta')^i\sqrt{\log x}}\frac{(1+2\eta')^i}{(\log\log x)(\log R_i)^2}\ll \frac{x}{\sqrt{\log x}(\log\log x)i^2}$$
and summing over $i$ shows that the total number of elements removed is $o(x/\sqrt{\log x})$.
}\\

\item[(ii)-(iv)]{cf. proof for $\cN_1$.
}\\

\item[(v)]{
First of all we have the bound (cf. \cite{KU14}) 
$$\sum_{n\in\cN(x)}f(n)\ll \frac{x}{\sqrt{\log x}}.$$
So Chebyshev gives that the number of $n\in\cN(x)$ s. t. $f(n)\geq T$ is $\ll x/(T\sqrt{\log x})$.

Let us denote $\cN_j=\{n\in\cN \mid f(n)\in[2^j,2^{j+1}]\}$. For each $n\in\cN_j$ let us remove $2^j/(j^2\log\log n)$ elements of $\cN$ to the right and left. 

Let us bound the total number of removed elements in $\cN(x)$ for the set $\cN_j$. It is sufficient to do this for elements $n\geq\sqrt{x}$, so that $\log\log n=(1+o(1))\log\log x$. Then, in view of Chebyshev, the total number of elements removed is
$$\ll \frac{x}{2^j\sqrt{\log x}}\cdot\frac{2^j}{j^2 \log\log x}
= \frac{x}{j^2\sqrt{\log x}\cdot \log\log x}$$
and summing over $j$ shows that the total number of elements removed for all sets $\cN_j$ is
$$\ll \frac{x}{\sqrt{\log x}\cdot\log\log x}=o(|\cN(x)|).$$
}
\end{itemize}

\subsection{Breakdown of multifractality} 

For rigid quantizations which have the property that the new eigenvalues stay close to the intermediate values of the old Laplace eigenvalues, in the sense that $$\ave{\delta_j}_x \asymp \ave{\Delta_j}_x = (\log x)^{1/2+o(1)},$$ we may see a breakdown of multifractality. 

We will use Lemma \ref{tail_bound} to obtain an estimate of the fractal exponent $d_q$ to illustrate this phenomenon.
Let $L=\log\tilde{n}_j$. The lemma yields the following estimate on the moment sum $m_q$:
\begin{equation}
\begin{split}
m_q(\tilde{n}_j)=~ &r_2(\tilde{n}_j) + \Delta_j^{2q}O(L^{(\log2/2-q+f(q)+o(1))})\\
=~ &L^{\log 2/2+o(1)}+O(L^{\log 2/2+f(q)+o(1)}).
\end{split}
\end{equation}

Thus, we find
$$\limsup_{n_j\in\cN'}\frac{\log m_q(n_j)}{\log\log n_j}\in[\tfrac{1}{2}\log 2,\tfrac{1}{2}\log 2+f(q)].$$

and (recall $\alpha=\tfrac{1}{2}$)
\begin{align*}
d_q=\limsup_{\tilde{n}_j\in\cN'\to+\infty}\frac{h_q(\tilde{n}_j)}{\log N}
=&2(1-\frac{1}{2q})\limsup_{\tilde{n}_j\in\cN'\to+\infty}\frac{\log m_q(\tilde{n}_j)}{\log\log \tilde{n}_j} \\ 
=&(1-\frac{1}{2q})(1+F(q))\log2
\end{align*}
where the function $F$ satisfies the inequality $$0\leq F(q)\leq \frac{2\log 2}{2q-1}$$ 
which does not rule out the possibility that $F(q)=(2q-1)^{-1}$ which in turn would yield $d_q=\log 2$. It is striking that even a precise asymptotic such as an analogue of Sathe-Selberg for integers representable as sums of two squares does not seem to be sufficient to improve this bound further. We are, therefore, led to suspect that multifractality may break down for such rigid quantizations.

\section{Multifractality of the ground state}\label{sec-ground}

Consider a general rectangular torus $\T^2=\R^2/2\pi\cL_0$, where $\cL_0=\Z(a,0)\oplus\Z(0,1/a)$ for some $a>0$. Denote the dual lattice of $\cL_0$ by $\cL$. We recall the $L^2$ expansion of the eigenfunction (we drop the $1/(4\pi)^2$ factor for convenience)
$$G_\lambda(x)=\sum_{\xi\in\cL}\frac{1}{|\xi|^2-\lambda}e^{i\left\langle\xi,x\right\rangle}$$

Let us consider the (non-normalized) complex moment
$$\zeta_\lambda(s)=\sum_{\xi\in\Z^2} \frac{1}{||\xi|^2-\lambda|^s}, \quad \re s> 1$$
which gives information about the multifractal properties of the eigenfunctions.

In order to study the multifractal properties of the ground state of the \u{S}eba billiard in the limit $\lambda\to 0$, we first of all note that the leading term is the coefficient $\lambda^{-s}$ which corresponds to the constant term (no dependence on $x$) in the $L^2$-expansion. 
We might expect interesting scaling properties of the fluctuations of the eigenfunction around the constant term (effectively, we are studying the scaling properties of the regularized function $G_\lambda+1/\lambda$).

Let us, therefore, study the modified moment
$$\zeta_\lambda^*(s)=\sum_{\xi\in\cL\setminus\{0\}}||\xi|^2-\lambda|^{-s}, \quad \re s> 1.$$
Accordingly, we define the $\ell^2$-normalized modified moment as
$$M_{s/2}^*(\lambda)=\frac{\zeta_\lambda^*(s)}{\zeta_\lambda^*(2)^{s/2}}.$$

The limit of $M_{s/2}^*(\lambda)$, as $\lambda\to 0$, exists and we have
$$M_{s/2}^*(0)=\frac{\zeta_Q(s)}{\zeta_Q(2)^{s/2}}$$
where $\zeta_Q$ denotes the Epstein zeta function for the quadratic form
$Q(m,n)=a^2m^2+a^{-2}n^2$: $$\zeta_Q(s)=\sum_{(m,n)\neq(0,0)}Q(m,n)^{-s}$$ 

In the half-plane $\re s> 1$, we have
$$M_{s/2}^*(\lambda)=\zeta_{Q}(s)/\zeta_{Q}(2)^{s/2}+\lambda F(s)+O(\lambda^2)\times G(s)$$
where $F(s)$ and $G(s)$ can be written explicitly in terms of $\zeta_{Q}$.

One may now continue $M_{s/2}^*(\lambda)$ analytically to the complex plane using the functional equation for $\zeta_{Q}$:
$$\zeta_Q(s)=\varphi_Q(s)\zeta_Q(1-s), \quad \varphi_Q(s)=\pi^{2s-1}\frac{\Gamma(1-s)}{\Gamma(s)}.$$ 

In order to derive the fractal exponent, we recall $q=s/2$ and we expect a scaling law of the form
$$M_q^*(\lambda)\sim k^{(1-q)D_q^*}, \quad \lambda\to0$$
where $k$ is the effective length of the summation (which is defined for $\re s>1$ and is constant, due to the fact
that we are dealing with the ground state and $\lambda$ is small).

Therefore, we define $$d_q^*
=\log \zeta_Q(2q)$$ as the fractal exponent with respect to the $\ell^\infty$-normalization
and, with respect to the $\ell^2$-normalization, we define
\begin{equation}
D_q^*
=\frac{\log \zeta_Q(2q)-q\log\zeta_Q(2)}{1-q}.
\end{equation}
where, for simplicity, we drop the constant factor $1/(\log k)$ (a normalization in some sense). 

Note that in the limit as $q\to1$ we obtain the Shannon entropy associated with the probability measure $\mu_\lambda(\xi)=\zeta_Q(2)^{-1}|\xi|^{-4}$:
$$\lim_{q\to 1}D_q^*=\log \zeta_Q(2)-2\frac{\zeta_Q'(2)}{\zeta_Q(2)}$$
which we may rewrite as 
$$H_{Sh.}(\mu_\lambda)=-\sum_{\xi\in\Z^2}\frac{1}{\zeta_Q(2)|\xi|^4}\log\left(\frac{1}{\zeta_Q(2)|\xi|^4}\right).$$ 


{\bf Symmetry around $q=\tfrac{1}{4}$:}
The functional equation for $\zeta_Q$ yields a similar property for $d_q^*$, $D_q^*$ with respect to the transformation $q\mapsto \tfrac{1}{2}-q$. 

We have 
$$d_{1/2-q}^*=d_q^*+\log\varphi_Q(2q)$$
and as an immediate consequence
\begin{equation}
\begin{split}
D_{1/2-q}^*=&\frac{d_{1/2-q}^*+\log\varphi_Q(2q)-(\tfrac{1}{2}-q)d_1^*}{\tfrac{1}{2}+q}\\
=&\frac{1-q}{\tfrac{1}{2}+q}\left(D_q^* +\frac{\log\varphi_Q(2q)+(2q-\tfrac{1}{2})\log \zeta_Q(2)}{1-q}\right).
\end{split}
\end{equation}

\section{Proof of Proposition \ref{scaling}}\label{sec-scaling}

We recall the definition $$\ave{f}_T=\frac{1}{T}\int_T^{2T}f(t)dt.$$

So let us calculate $\ave{\tau_q}_T$. We define the summation into two ranges: $|m-t|\leq T$ and its complement in $\cN$.

Let us denote $$\tau_q^t(t,G)=\sum_{G\leq|m-t|\leq T} r_2(m)|m-t|^{-2q}$$
and denote the sum over the complement $\tau_q^r(t,G)$.

First of all note that the contribution from $\tau_q^r$ is negligible (note $|t|\leq T$:
\begin{align*}
\sum_{|m-t|>T} r_2(m)|m-t|^{-2q}\ll_\epsilon \int_{|x-t|>T}\frac{x^\epsilon}{|x-t|^{2q}}dx\\
=\int_{|s|>T}\frac{|s+t|^\epsilon}{|s|^{2q}}ds=O_\epsilon(T^{1-2q+\epsilon}).
\end{align*}

Let us now calculate the mean of $\tau_q^t$:
\begin{equation}
\ave{\tau_q^t}_T = \frac{1}{T}\sum_{m\in[0,3T]} r_2(m) \int_T^{2T} |t-m|^{-2q}~\mathbbm{1}(t\mid G\leq |t-m|\leq T) dt
\end{equation}
where the summation is restricted to $m\in[0,3T]$, because for $m$ outside this interval we have $[T,2T]\cap[m-T,m+T]=\emptyset$.

The main contribution will arise from terms $m\in[T,2T]$. 
We will deal with these terms first and rewrite the integral inside the sum as
$$I_m=\int_{T-m}^{2T-m}|\tau|^{-2q}~\mathbbm{1}(\tau \mid G\leq |\tau|\leq T) d\tau$$

We may evaluate (recall $m\geq T$)
$$I_m^-=\int_{T-m}^0 |\tau|^{-2q}~\mathbbm{1}(\tau \mid G\leq |\tau|\leq T) d\tau
=
\begin{cases}
\int_{T-m}^{-G} |\tau|^{-2q} d\tau, \quad \text{if}~m\geq T+G, \\
\\
0, \quad \text{otherwise.}
\end{cases}
$$
which yields
$$
I_m^-=
\begin{cases}
\frac{1}{1-2q}(-G^{1-2q}+(m-T)^{1-2q}), \quad \text{if}~m\geq T+G, \\
\\
0, \quad \text{otherwise.}
\end{cases}
$$

Similarly, we may evaluate
\begin{equation}
\begin{split}
I_m^+=&\int^{2T-m}_0 \tau^{-2q}~\mathbbm{1}(\tau \mid G\leq |\tau|\leq T) d\tau\\
=&
\begin{cases}
\frac{1}{1-2q}((2T-m)^{1-2q}-G^{1-2q}), \quad \text{if}~ m\leq 2T-G,\\
\\
0, \quad \text{otherwise.}
\end{cases}
\end{split}
\end{equation}

Let us now calculate the sums over $m$.

We have
\begin{equation}
\begin{split}
\sum_{m\in[T+G,2T]}r_2(m)I_m^-
=&~\frac{1}{2q-1}G^{1-2q}\sum_{m\in[T+G,2T]}r_2(m)\\
&~+\frac{1}{1-2q}\sum_{m\in[T+G,2T]}r_2(m)(m-T)^{1-2q}
\end{split}
\end{equation}
and
\begin{equation}
\begin{split}
\sum_{m\in[T,2T-G]}r_2(m)I_m^+
=&~\frac{1}{2q-1}G^{1-2q}\sum_{m\in[T,2T-G]}r_2(m)\\
&~-\frac{1}{2q-1}\sum_{m\in[T,2T-G]}r_2(m)(2T-m)^{1-2q}
\end{split}
\end{equation} 
and we estimate
$$\sum_{m\in[T+G,2T]}r_2(m)(m-T)^{1-2q}\ll_\epsilon \int_{T+G}^{2T} t^\epsilon(t-T)^{1-2q}dt
=\int_G^T (t+T)^\epsilon t^{1-2q}dt$$
which is $O_\epsilon((G+T)^\epsilon G^{2-2q})$.

Similarly, one may show that $$\sum_{m\in[T,2T-G]}r_2(m)(2T-m)^{1-2q}=O_\epsilon((G+T)^\epsilon G^{2-2q}).$$

So we find
$$\frac{1}{T}\sum_{m\in[T,2T]} r_2(m)I_m =\frac{2\pi}{2q-1}G^{1-2q}(1+O(T^{-1+\theta}))$$
where $\theta$ denotes the exponent in the circle law:
$$\sum_{m\leq T}r_2(m)=\pi T+O(T^\theta)$$ 

The terms $m\in[0,T]\cup[2T,3T]$ give a contribution of lower order.

To see this consider $m\in[0,T]$ (the other case is very similar). 

We have
$$\int_{T-m}^{2T-m}\tau^{-2q}~\mathbbm{1}(\tau \mid G\leq |\tau|\leq T)d\tau$$
and since $m\leq T$ implies $[T-m,2T-m]\subset\R_+$, we can evaluate this integral as
$$
=
\begin{cases}
\int_{T-m}^{T}\tau^{-2q}d\tau, \quad \text{if}\; m\leq T-G,\\
\\
\int_G^{T}\tau^{-2q}d\tau, \quad \text{if}\; m>T-G.
\end{cases}
$$

Therefore the main contribution is of the form
$$
\frac{1}{2q-1}G^{1-2q}\frac{1}{T}\sum_{m\in[T-G,T]}r_2(m)= \frac{1}{2q-1}G^{1-2q}(\pi G/T+O(T^{-1+\theta})).
$$
and we recall $G=o(T)$.


\begin{thebibliography}{99}  
\bibitem{Al86}
B. L. Altshuler, V. E. Kravtsov, I. V. Lerner, {\em Statistical properties of mesoscopic fluctuations and similarity theory},
JEPT Lett. 43 (1986), no. 7, 441--44.
\bibitem{A08}
N. Anantharaman, {\em Entropy and the localization of eigenfunctions},
Ann. Math. (2) 168 (2008), no. 2, 435--475.
\bibitem{AtBo12}
Y. Y. Atas and E. Bogomolny, {\em Multifractality of eigenfunctions in spin chains},
Phys. Rev. E 86 (2012), 021104.
\bibitem{BHK19}
A. B\"acker, M. Haque, and I. M. Khaymovich, {\em Multifractal dimensions for random matrices, chaotic quantum maps, and many-body systems}, 
Phys. Rev. E 100 (2019), 032117.
\bibitem{BeBoK01}
G. Berkolaiko, E. B. Bogomolny, J. P. Keating, {\em Star graphs and Seba billiards}, 
J. Phys. A 34 (2001), 335--350.
\bibitem{BK99}
G. Berkolaiko, J. P. Keating, {\em Two-point spectral correlations for star graphs}, 
J. Phys. A 32 (1999), 7827--7841.
\bibitem{BKW03}
G. Berkolaiko, J. P. Keating, B. Winn, {\em Intermediate wave function statistics}, 
Phys. Rev. Lett. 91 (2003), 134103.
\bibitem{BKW04}
G. Berkolaiko, J. P. Keating, B. Winn, {\em No quantum ergodicity for star graphs}, 
Commun. Math. Phys. 250 (2004), 259--285.
\bibitem{Gi20}
A. M. Bilen, I. Garcia-Mata, B. Georgot, O. Giraud,
{\em Multifractality of open quantum systems},
Phys. Rev. E 100 (2020), 032223.
\bibitem{BoGeSch01}
E. Bogomolny, U. Gerland, C. Schmit, {\em Singular Statistics}, Phys. Rev. E 63 (2001), 036206.
\bibitem{BoGi11}
E. Bogomolny, O. Giraud,
{\em Perturbation approach to multifractal dimensions for certain critical random matrix ensembles},
Phys. Rev. E 84 (2011), 036212.
\bibitem{BoLeSch00}
E. Bogomolny, P. Leboeuf, C. Schmit, {\em Spectral statistics of chaotic systems with a pointlike scatterer}, 
Phys. Rev. Lett. 85 (2000), 2486--2489.
\bibitem{CdV85} 
Y. Colin de Verdi\`ere, {\em Ergodicit\'e et fonctions propres du laplacien}, 
Comm. Math. Phys. 102 (1985), no. 3, 497--502.
\bibitem{FNdB03} 
F. Faure, S. Nonnenmacher, S. de Bi\`evre, {\em Scarred eigenstates for quantum
cat maps of minimal periods}, Comm. Math. Phys. 239 (2003), no. 3, 449--492.
\bibitem{Ha10}
A. Hassell, {\em Ergodic billiards that are not quantum unique ergodic (with an appendix by Andrew Hassell and Luc Hillairet)}
Ann. Math. (2) 171 (2010), no. 1, 605--618.
\bibitem{Hu03}
M. N. Huxley, {\em Exponential sums and lattice points. III.}, Proc. London Math. Soc. 87 (2003), No. 3, 591--609.
\bibitem{Ho10a}
R. Holowinsky, {\em Sieving for mass equidistribution},
Ann. Math. (2) 172 (2010), no. 2, 1499--1516.
\bibitem{Ho10b}
R. Holowinsky, K. Soundararajan, {\em Mass equidistribution for Hecke eigenforms},
Ann. Math. (2) 172 (2010), no. 2, 1517--28.
\bibitem{KMW03}
J. P. Keating, J. Marklof and B. Winn, {\em Value distribution of the eigenfunctions and spectral determinants of quantum star graphs}, 
Comm. Math. Phys. 241 (2003), 421--452.
\bibitem{KM97}
V. E. Kravtsov, K. A. Muttalib, {\em New Class of Random Matrix Ensembles with Multifractal Eigenvectors}, Phys. Rev. Lett. 79 (1997), 1913--1916.
\bibitem{KRo17}  
P. Kurlberg, L. Rosenzweig, {\em Superscars for Arithmetic Toral Point Scatterers},
Comm. Math. Phys. 349 (2017), no. 1, 329--360.
\bibitem{KLeRo20}
P. Kurlberg, S. Lester, L. Rosenzweig, {\em Superscars for Arithmetic Toral Point Scatterers II}, preprint.
\bibitem{KU14}  
P. Kurlberg, H. Uebersch\"ar, {\em Quantum Ergodicity for Point Scatterers on Arithmetic Tori},
Geom. Funct. Anal. 24 (2014), no. 5, 1565--1590.
\bibitem{KU17} 
P. Kurlberg, H. Uebersch\"ar, {\em Superscars in the \u{S}eba billiard},
J. Eur. Math. Soc. 19 (2017), no. 10, pp. 2947--2964.
\bibitem{L06} 
E. Lindenstrauss, {\em Invariant measures and arithmetic quantum unique ergodicity}, 
Ann. Math. (2) 163 (2006), no. 1, 165--219. 
\bibitem{MR00}
J. Marklof, Z. Rudnick, {\em Quantum unique ergodicity for parabolic maps},
Geom. Funct. Anal. 10 (2000) 1554--1578
\bibitem{Gi10}
J. Martin, I. Garcia-Mata, O. Giraud, and B. Georgeot,
{\em Multifractal wave functions of simple quantum maps},
 Phys. Rev. E 82 (2010), 046206.
\bibitem{Mi00}
A. D. Mirlin, {\em Statistics of energy levels and eigenfunctions in disordered systems},
Physics Reports 326 (2000), 259--382.
\bibitem{RB81} 
P. J. Richens, M. V. Berry, {\em Pseudointegrable systems in classical and quantum mechanics}, Physica D 2 (1981), 495--512.
\bibitem{RS94} 
Z. Rudnick, P. Sarnak, {\em The behaviour of eigenstates of arithmetic hyperbolic manifolds}, 
Comm. Math. Phys 161 (1994), no. 1, 195--213.
\bibitem{RU12}
Z. Rudnick, H. Uebersch\"ar, {\em Statistics of Wave Functions for a Point Scatterer on the Torus},
Comm. Math. Phys. 316 (2012), no. 3, 763--782.
\bibitem{RU14}
Z. Rudnick, H. Uebersch\"ar, {\em On the Eigenvalue Spacing Distribution for a Point Scatterer on the Flat Torus}, Ann. Henri Poincar\'e 15 (2014), 1--27.
\bibitem{SG91}
M. Schreiber, H. Grussbach, {\em Multifractal wave functions at the Anderson transition}
Phys. Rev. Lett. 67 (1991), 607--610.
\bibitem{Se90}  
P. \u{S}eba, {\em Wave chaos in singular quantum billiard}, Phys. Rev. Lett. 64 (1990), 1855--1858.
\bibitem{Sh48a}
C. E. Shannon, {\em A mathematical theory of communication}, The Bell System Technical Journal 27 (1948), no. 3, 379--423.
\bibitem{Sh48b}
C. E. Shannon, {\em A mathematical theory of communication}, The Bell System Technical Journal 27 (1948), no. 4, 623--666.
\bibitem{Sh94}
T. Shigehara, {\em Conditions for the appearance of wave
chaos in quantum singular systems with a pointlike scatterer}, Phys Rev. E 50 (1994), no.6, 4357--4370.
\bibitem{Sn74}
A. I. Snirelman, {\em Ergodic properties of eigenfunctions}, Uspehi Mat. Nauk. 29
(1974), no. 6 (180), 181--182.
\bibitem{So10}
K. Soundararajan, {\em Quantum Unique Ergodicity for SL(2,Z)\textbackslash H},
Ann. Math. (2) 172 (2010), no. 2, pp. 1529--1538.
\bibitem{U14}
H. Uebersch\"ar, {\em Quantum chaos for point scatterers on flat tori}, Phil. Trans. R. Soc. A.372 (2014), 20120509.
\bibitem{Z87}
S. Zelditch, {\em Uniform distribution of eigenfunctions on compact hyperbolic surfaces},
Duke Math. J. 55 (1987), no. 4, 919--941.

\end{thebibliography}
\end{document}